\documentclass[journal, draftclsnofoot,oneside, onecolumn,11pt, letterpaper]{IEEEtran}

\ifCLASSINFOpdf
\else
\fi
\usepackage{pslatex}
\usepackage{amsfonts,color,morefloats}
\usepackage{amssymb,amsmath,latexsym,amsthm}
\usepackage{bbm}
\usepackage{amsfonts}
\usepackage{mathrsfs}
\usepackage{amsthm}
\usepackage{latexsym}
\usepackage{color}
\usepackage{dcolumn}
\usepackage{extarrows}
\usepackage{amsmath}

\allowdisplaybreaks[4]

\newtheorem{fact}{Fact}
\newtheorem{theorem}{Theorem}
\newtheorem{lemma}[theorem]{Lemma}

\newtheorem{corollary}[theorem]{Corollary}

\newtheorem{definition}{Definition}
\newtheorem{example}{Example}
\newtheorem{remark}{Remark}
\newtheorem{case}{Case}

\begin{document}
%
\title{Linear Complexity of Binary Interleaved Sequences of Period 4$n$\thanks{The work was supported by the National Natural Science Foundation of China (NSFC) under Grant 12031011.}}
%
%
%

\author{Qiuyue~Liu,
        Shiyuan~Qiang,
        Minghui~Yang,
        and~Keqin~Feng
\thanks{Qiuyue Liu is with the Center for Applied Mathematics, Tianjin University, Tianjin 300072, China (e-mail: 17864309562@163.com).}
\thanks{Shiyuan Qiang is with the Department of Applied Mathematics, China Agricultural University, Beijing 100083, China (e-mail: qsycau\_18@163.com).}
\thanks{Minghui Yang is with State Key Laboratory of Information Security, Institute of Information Engineering, Chinese Academy of Sciences, Beijing 100093, China (e-mail:  yangminghui6688@163.com).}
\thanks{Keqin Feng is with the Department of Mathematical Sciences, Tsinghua University, Beijing 100084, China (e-mail: fengkq@tsinghua.edu.cn).}
}

%
%

\markboth{}%
{Shell \MakeLowercase{\textit{et al.}}: Bare Demo of IEEEtran.cls for IEEE Journals}
%



\maketitle

\begin{abstract}
Binary periodic sequences with good autocorrelation property have many applications in many aspects of communication. In past decades many series of such binary sequences have been constructed. In the application of cryptography, such binary sequences are required to have larger linear complexity. Tang and Ding \cite{X. Tang} presented a method to construct a series of binary sequences with period 4$n$ having optimal autocorrelation. Such sequences are interleaved by two arbitrary binary sequences with period $n\equiv 3\pmod 4$ and ideal autocorrelation. In this paper we present a general formula on the linear complexity of such interleaved sequences. Particularly, we show that the linear complexity of such sequences with period 4$n$ is not bigger than $2n+2$. Interleaving by several types of known binary sequences with ideal autocorrelation ($m$-sequences, Legendre, twin-prime and Hall's sequences), we present many series of such sequences having the maximum value $2n+2$ of linear complexity which gives an answer of a problem raised by N. Li and X. Tang \cite{N. Li}. Finally, in the conclusion section we show that it can be seen easily that the 2-adic complexity of all such interleaved sequences reaches the maximum value $\log_{2}(2^{4n}-1)$.
\end{abstract}


%
\IEEEpeerreviewmaketitle

\section{Introduction}

Let
 $a=(a_i)_{i=0}^{n-1}$ be a binary sequence with period $n$, $a_i\in \{0, 1\}=\mathbb{F}_{2}$. The autocorrelation of $a$ is defined by
$$A_a(\tau)=\sum_{i=0}^{n-1}(-1)^{a_i+a_{i+\tau}} \ (0\leq\tau\leq n-1).$$

Binary sequence $a$ with small value $\max\{|A_a(\tau)|: 1\leq\tau\leq n-1\}$ has been used in many aspects of communication ([1, 4]). For $n\equiv3\pmod 4$, $n\geq3$, a binary sequence $a$ with period $n$ is called to have ideal autocorrelation if $A_a(\tau)=-1$ for all $\tau$, $1\leq\tau\leq n-1$. Many series of binary sequences with ideal autocorrelation have been constructed ($m$-sequences, Legendre sequences, Hall's sequences and many others as in \cite{J.-S. No, T. Storer}).

In the application of cryptography, a binary sequence is also required to have larger linear complexity.

\begin{definition}\label{def1}
For a binary sequence $a=(a_i)_{i=0}^{N-1} \ (a_i\in \{0, 1\})$ with period $N$, the linear complexity of $a$ is defined by $LC(a)=N-d$ where
$$d=\deg f_{a}(x), ~f_{a}(x)=\gcd (S_{a}(x), x^{N}-1)~~\textrm{(the greatest common divisor in $\mathbb{F}_{2}[x])$}$$
and $S_{a}(x)=\sum_{i=0}^{n-1}a_{i}x^{i}\in \mathbb{F}_{2}[x]$.
\end{definition}

Tang and Ding \cite{X. Tang} presented a method to construct a series of binary sequences $w=w(a, b)=(w_i)_{i=0}^{4n-1}$ with period 4$n$ and optimal autocorrelation property, which means that $A_w(\tau)\in \{0, -4\}$ for all $1\leq\tau\leq 4n-1$. Such sequences are interleaved by two arbitrary binary sequences $a$ and $b$ with ideal autocorrelation. In this paper we determine the linear complexity of such interleaved sequences.

In Section \ref{sec2} we introduce the binary interleaved sequences constructed by Tang and Ding in \cite{X. Tang}. In Section \ref{sec3} we present a general formula on the linear complexity $LC(w)$ of such interleaved sequence $w$ with period 4$n$~(Theorem \ref{th2}) and show that $LC(w)\leq 2n+2$. Li and Tang \cite{N. Li} raised a problem: to find interleaved sequence $w$ such that $LC(w)=2n+2$. As a direct consequence of Theorem \ref{th2} we present a necessary and sufficient condition on $LC(w)=2n+2$, and then we show many series of such sequences with the maximum linear complexity by interleaving several known types of binary sequences with ideal autocorrelation. In the conclusion section we remark that it is easy to see that the 2-adic complexity of all such interleaved sequences with period 4$n$ reaches the maximum value $\log_{2}(2^{4n}-1)$.

\section{Binary Interleaved Sequences}\label{sec2}

In this section we introduce the binary interleaved sequences presented in \cite{X. Tang}. For general theory on interleaved sequences we refer to G. Gong \cite{G. Gong}.

Let $n=4m-1~(m\geq1)$, $A=(A_{i})$, $B=(B_{i})$, $C=(C_{i})$ and $D=(D_{i})$ be binary sequences with period $n$, $A_{i}, B_{i}, C_{i}, D_{i}\in \{0, 1\}$. Considering the following  $n\times4$ $\{0, 1\}$-matrix
$$\left(
  \begin{array}{cccc}
    A_{0} & B_{0} & C_{0} & D_{0} \\
    A_{1} & B_{1} & C_{1} & D_{1} \\
    \vdots & \vdots & \vdots & \vdots \\
    A_{n-1} & B_{n-1} & C_{n-1} & D_{n-1} \\
  \end{array}
\right),
$$
the binary sequence interleaved by $A, B, C, D$ is the sequence with period 4$n$ defined by
\begin{align}
w& =I(A, B, C, D)=(w_{0}, w_{1},\cdots, w_{4n-1})\notag\\
                                  &= (A_{0}, B_{0}, C_{0}, D_{0}, A_{1}, B_{1}, C_{1}, D_{1}, \cdots, A_{n-1}, B_{n-1}, C_{n-1}, D_{n-1})\notag.
                                 \end{align}

Since $n=4m-1$ is odd, we have the following isomorphism of rings by the Chinese Remainder Theorem:
$$\varphi: \mathbb{Z}_{4n}\cong \mathbb{Z}_{4}\oplus \mathbb{Z}_{n},  ~~~\varphi(i (\bmod {4n}))=(i (\bmod 4), i (\bmod n)).$$
The pre-image of $(\alpha, \beta)\in \mathbb{Z}_{4}\oplus \mathbb{Z}_{n}$ is $\varphi^{-1}(\alpha, \beta)=-\alpha n+4\beta^{*}$ where $\beta^{*}\in \mathbb{Z}$ satisfying $4\beta^{*}\equiv\beta\pmod n$, $0\leq \beta\leq n-1$. From the definition of $w$ we know that for $i=-\alpha n+4\beta^{*}$ we have
\begin{align}\label{e1}
 w_{i}
 &= \left\{ \begin{array}{ll}
A_{\beta^{*}}, & \textrm{if $\alpha=0$};\\
B_{\frac{1}{4}(-n+4\beta^{*}-1)=B_{\beta^{*}}-\frac{n+1}{4}=B_{\beta^{*}-m}}, & \textrm{if $\alpha=1$};\\
C_{\frac{1}{4}(-2n+4\beta^{*}-2)=C_{\beta^{*}}-\frac{2(n+1)}{4}=C_{\beta^{*}-2m}}, & \textrm{if $\alpha=2$};\\
D_{\beta^{*}-3m}, & \textrm{if $\alpha=3$}.
\end{array} \right.
\end{align}

The binary interleaved sequence considered in \cite{X. Tang} is
$$w=w(a, b)=I(a, L^{m}(b), L^{2m}(a), L^{3m}(\overline{b}))$$
where $a$ and $b$ are arbitrary binary sequences with period $n=4m-1$ and ideal autocorrelation, $L^{s}(a)=(\ell_{i})_{i=0}^{n-1}$ is the shift sequence of $a$ defined by $\ell_{i}=a_{i+s}$, and $\overline{b}=(\overline{b_{i}})_{i=0}^{n-1}$ is the complementary sequence of $b$ defined by $\overline{b}_{i}=1-b_{i}$. It is proved in \cite{X. Tang}, all such interleaved sequences $w=w(a, b)$ have optimal autocorrelation: $A_{w}(\tau)=0$ or $-4$ for all $1\leq \tau\leq 4n-1$. In order to determine the linear complexity of $w=w(a, b)$, we need to compute the value $S_{w}(x)=\sum_{i=0}^{4n-1}w_{i}x^{i}\in \mathbb{F}_{2}[x]\pmod {x^{4n}-1}$.

\begin{lemma}\label{lem1}
Let $a=(a_i)_{i=0}^{n-1}$, $b=(b_i)_{i=0}^{n-1}$ be binary sequences with period $n=4m-1(m\geq1)$, $a_i, b_i\in \{0, 1\}$, $w=(w_{i})_{i=0}^{4n-1}$ be the binary interleaved sequence $w=w(a, b)=I(a, L^{m}(b), L^{2m}(a), L^{3m}(\overline{b}))$ with period 4$n$, $S_{w}(x)=\sum_{i=0}^{4n-1}w_{i}x^{i}\in \mathbb{F}_{2}[x]$. Then in $\mathbb{F}_{2}[x]$, we have
$$S_{w}(x)\equiv(1+x^{2n})\sum_{i=0}^{n-1}a_{i}x^{4i}+(x^{n}+x^{3n})\sum_{i=0}^{n-1}b_{i}x^{4i}+x^{3}\cdot \frac{x^{4n}-1}{x^{4}-1}\pmod {x^{4n}-1}.$$
\end{lemma}

\begin{proof}
For general interleaved sequence $w=I(A, B, C, D)=(w_{i})_{i=0}^{4n-1}$,
\begin{align*}
S_{w}(x) &\equiv \sum_{\alpha=0}^{3}\sum_{\beta=0}^{n-1}w_{-\alpha n+4\beta^{*}}x^{-\alpha n+4\beta^{*}}\pmod{x^{4n}-1}\notag\\
                                  &\equiv \sum_{\alpha=0}^{3}\sum_{\beta=0}^{n-1}w_{-\alpha n+4\beta}x^{-\alpha n+4\beta}\pmod{x^{4n}-1}\notag.
\end{align*}
(From $4\beta^{*}\equiv \beta\pmod n$ and $2\nmid n$ we know that $\beta\rightarrow\beta^{*}$ is a permutation on $\mathbb{Z}_{n}$.)\\
Then by formula (\ref{e1}) in Section II,
\begin{align*}
S_{w}(x) &\equiv \sum_{\beta=0}^{n-1}(A_{\beta}x^{4\beta}+B_{\beta-m}x^{-n+4\beta}
         +C_{\beta-2m}x^{-2n+4\beta}+D_{\beta-3m}x^{-3n+4\beta})\pmod{x^{4n}-1}\notag\\
         &\equiv \sum_{\beta=0}^{n-1}(A_{\beta}x^{4\beta}+B_{\beta-m}x^{1+4(\beta-m)}
         +C_{\beta-2m}x^{2+4(\beta-2m)}+D_{\beta-3m}x^{3+4(\beta-3m)})\pmod{x^{4n}-1}\notag\\
         &\ \ \ (\text{since} ~n=4m-1)\notag\\
         &\equiv \sum_{i=0}^{n-1}(A_{i}x^{4i}+B_{i}x^{1+4i}
         +C_{i}x^{2+4i}+D_{i}x^{3+4i})\pmod{x^{4n}-1}\notag\\
         &\equiv \sum_{i=0}^{n-1}x^{4i}(A_{i}+B_{i}x
         +C_{i}x^{2}+D_{i}x^{3})\pmod{x^{4n}-1}\notag.
\end{align*}
Now let $(A, B, C, D)=(a, L^{m}(b), L^{2m}(a), L^{3m}(\overline{b}))$. Then $w=w(a, b)$ and
\begin{align*}
S_{w}(x)& \equiv \sum_{i=0}^{n-1}x^{4i}(a_{i}+b_{i+m}x+a_{i+2m}x^{2}+(1+b_{i+3m})x^{3})\pmod{x^{4n}-1}\notag\\
&\equiv (1+x^{2n})\sum_{i=0}^{n-1}a_{i}x^{4i}+(x^{n}+x^{3n})
\sum_{i=0}^{n-1}b_{i}x^{4i}+x^{3}\cdot\frac{x^{4n}-1}{x^{4}-1}\pmod{x^{4n}-1}.
\end{align*}
\end{proof}

\section{Linear Complexity: General Result}\label{sec3}

Recall that for a binary sequence $w=(w_{i})_{i=0}^{N-1}$ with period $N$, the linear complexity of $w$ is defined by
$$LC(w)=N-d, d=\deg f_{w}(x), f_{w}(x)=\gcd (S_{w}(x), x^{N}-1)\in \mathbb{F}_{2}[x], S_{w}(x)=\sum_{i=0}^{N-1}w_{i}x^{i}\in \mathbb{F}_{2}[x].$$
and $LC(w)$ reaches the maximum value $N$ if and only if $f_{w}(x)=1$.

In this section we show the following general result on the linear complexity of the binary interleaved sequence $w=w(a, b)$ where $a$ and $b$ are binary sequences with period $n=4m-1(m\geq1)$ and ideal autocorrelation.

\begin{theorem}\label{th2}
Let $n=4m-1(m\geq1)$, $a=(a_i)_{i=0}^{n-1}$ and $b=(b_i)_{i=0}^{n-1}$ be binary sequences with period $n$ and ideal autocorrelation. Let $\theta$ be an $n$-th primitive root of 1 in an extension field of $\mathbb{F}_{2}$ (such $\theta$ exists since $2\nmid n$). Let
$$S_{a}(x)=\sum_{i=0}^{n-1}a_{i}x^{i}, ~~S_{b}(x)=\sum_{i=0}^{n-1}b_{i}x^{i}, ~~S_{a+b}(x)=S_{a}(x)+S_{b}(x)~~(in ~\mathbb{F}_{2}[x]),$$
$$Z_{a}=\{1\leq \lambda\leq n-1: S_{a}(\theta^{\lambda})=0\}, ~~Z_{b}=\{1\leq \lambda\leq n-1: S_{b}(\theta^{\lambda})=0\},$$
$$Z_{a+b}=\{1\leq \lambda\leq n-1: S_{a+b}(\theta^{\lambda})=0\}.$$
Then the linear complexity of the interleaved sequence $w=w(a, b)$ is
\begin{align*}
LC(w)& =2n+2-|Z_{a}\cap Z_{b}|-|Z_{a+b}|\notag\\
& =2n+2-\deg\gcd \left(S_{a}(x), S_{b}(x), \frac{x^{n}-1}{x-1}\right)-\deg\gcd \left(S_{a+b}(x), \frac{x^{n}-1}{x-1}\right).
\end{align*}
\end{theorem}

\begin{proof}
From $2\nmid n$ we know that
$$x^{4n}-1=(x^{n}-1)^{4}=\prod_{\lambda=0}^{n-1}(x-\theta^{\lambda})^{4}
=(x-1)^{4}\prod_{\lambda=1}^{n-1}(x-\theta^{\lambda})^{4}$$
By Lemma \ref{lem1} we have
$$S_{w}(x)=\sum_{i=0}^{4n-1}w_{i}x^{i}\equiv\left(\frac{x^{n}-1}{x-1}\right)^{2}
\left[(x-1)^{2}S_{a}(x)^{4}+x^{n}(x-1)^{2}S_{b}(x)^{4}+x^{3}\left(\frac{x^{n}-1}{x-1}\right)^{2}\right]\pmod {(x^{n}-1)^{4}}.$$
Therefore
$$f_{w}(x)=\gcd (S_{w}(x), x^{4n}-1)=\left(\frac{x^{n}-1}{x-1}\right)^{2}\cdot \gcd (F(x), (x^{n}-1)^{2}(x-1)^{2})$$
where
$$F(x)=(x-1)^{2}S_{a}(x)^{4}+x^{n}(x-1)^{2}S_{b}(x)^{4}+x^{3}\left(\frac{x^{n}-1}{x-1}\right)^{2}\in\mathbb{F}_{2}[x]. $$
Since $F(x)\equiv x^{3}\left(\frac{x^{n}-1}{x-1}\right)^{2}\equiv n^{2}x^{3}\equiv1\pmod {x-1}$, we know that $\gcd (F(x), x-1)=1$ and $f_{w}(x)=\left(\frac{x^{n}-1}{x-1}\right)^{2}h(x)$ where $h(x)=\gcd \left(F(x), \left(\frac{x^{n}-1}{x-1}\right)^{2}\right)=\gcd \left(g(x), \left(\frac{x^{n}-1}{x-1}\right)^{2}\right)$ and $g(x)=S_{a}(x)^{4}+x^{n}S_{b}(x)^{4}\in \mathbb{F}_{2}[x]$. The derivative of $g(x)$ in $\mathbb{F}_{2}[x]$ is $g'(x)=nx^{n-1}S_{b}(x)^{4}=x^{n-1}S_{b}(x)^{4}$. On the other hand, $\left(\frac{x^{n}-1}{x-1}\right)^{2}=\prod_{\lambda=1}^{n-1}(x-\theta^{\lambda})^{2}$. Therefore $h(x)=\prod_{\lambda=1}^{n-1}(x-\theta^{\lambda})^{h_{\lambda}}$, $0\leq h_{\lambda}\leq2$. For $1\leq \lambda\leq n-1$, we have
$$g(\theta^{\lambda})=S_{a}(\theta^{\lambda})^{4}+\theta^{\lambda n}S_{b}(\theta^{\lambda})^{4}=S_{a+b}(\theta^{\lambda})^{4}.$$

(I). If $S_{a+b}(\theta^{\lambda})\neq0$, then $g(\theta^{\lambda})\neq0$ and $h(\theta^{\lambda})\neq0$, we get $h_{\lambda}=0$.

(II). If $S_{a+b}(\theta^{\lambda})=0$ and $S_{b}(\theta^{\lambda})\neq0$, then $S_{a}(\theta^{\lambda})\neq0$ and $g(\theta^{\lambda})=0$,  $g'(\theta^{\lambda})\neq0$. We get $h_{\lambda}=1$.

(III). If $S_{a+b}(\theta^{\lambda})=S_{b}(\theta^{\lambda})=0$, then $S_{a}(\theta^{\lambda})=0$ and $(x-\theta^{\lambda})^{2}\mid g(x)$. We get $h_{\lambda}=2$.\\
Then we get
$$h(x)=\prod_{\lambda=1 \atop S_{a}(\theta^{\lambda})=S_{b}(\theta^{\lambda})=0}^{n-1}(x-\theta^{\lambda})^{2}\cdot \prod_{\lambda=1 \atop S_{a}(\theta^{\lambda})=S_{b}(\theta^{\lambda})\neq0}^{n-1}(x-\theta^{\lambda})$$
And
\begin{align*}
LC(w)& =4n-\deg \gcd (S_{w}(x), x^{4n}-1)=4n-\deg f_{w}(x)\notag\\
& =4n-2(n-1)-\deg h(x)=2n+2-2|Z_{a}\cap Z_{b}|-|Z_{a+b}\backslash (Z_{a}\cap Z_{b})|\notag\\
& =2n+2-|Z_{a}\cap Z_{b}|-|Z_{a+b}|.
\end{align*}
Since $Z_{a}\cap Z_{b}=\{1\leq \lambda\leq n-1: S_{a}(\theta^{\lambda})=S_{b}(\theta^{\lambda})=0\}$, $Z_{a+b}=\{1\leq \lambda\leq n-1: S_{a+b}(\theta^{\lambda})=0\}$, we know that $|Z_{a}\cap Z_{b}|=\deg\gcd \left(S_{a}(x), S_{b}(x), \frac{x^{n}-1}{x-1}\right)$ and $|Z_{a+b}|=\deg\gcd \left(S_{a+b}(x), \frac{x^{n}-1}{x-1}\right)$. This completes the proof of Theorem \ref{th2}.
\end{proof}

Theorem \ref{th2} has several consequences. Firstly we know that $LC(w)\leq 2n+2$. Moreover, $LC(w)=2n+2$ if and only if $Z_{a}\cap Z_{b}$ and $Z_{a+b}$ are empty. But $Z_{a}\cap Z_{b}\subseteq Z_{a+b}$, we get the following consequence.

\begin{corollary}\label{co3}
Let $a$ and $b$ be binary sequences with period $n\equiv3\pmod 4$ and ideal autocorrelation. Then the linear complexity of $w=w(a, b)$ reaches the maximum value $2n+2$ if and only if $Z_{a+b}=\emptyset$, namely, for each $\lambda (1\leq \lambda\leq n-1)$, $S_{a}(\theta^{\lambda})+S_{b}(\theta^{\lambda})\neq0$.
\end{corollary}
From now on we denote $\sum(n)$ the set of all binary sequences with period $n$, $\sum'(n)$ the set of all binary sequences with period $n\equiv3\pmod 4$ and ideal autocorrelation. We introduce a transformation group $G$ on $\sum'(n)$, so that from one sequence $a$ in $\sum'(n)$ we can get more sequences $\sigma(a)(\sigma\in G)$ in $\sum'(n)$.

Let $a=(a_i)_{i=0}^{n-1}\in \sum(n)$.

\begin{definition} \label{def2}
(1). The complementary sequence of $a$ is $\overline{a}=(\overline{a_i})_{i=0}^{n-1}\in \sum(n)$ defined by $\overline{a_{i}}=1-a_{i}(0\leq i\leq n-1)$.

(2). For $0\leq r\leq n-1$, the shift sequence $L^{r}(a)=(c_i)_{i=0}^{n-1}$ of $a$ is defined by $c_{i}=a_{i+r} (0\leq i\leq n-1)$. We have $L^{r}(a)\in \sum(n)$, $L^{0}(a)= a$ and $L^{r}L^{r'}=L^{r+r'}$.

(3). For $s\in \mathbb{Z}_{n}^{*}$ (namely, $1\leq s\leq n-1$, $\gcd(s, n)=1$), the sample sequence $M_{s}(a)=(b_{i})_{i=0}^{n-1}$ of $a$ is defined by $b_{i}=a_{si}(0\leq i\leq n-1)$. We have $M_{s}(a)\in \sum(n)$,  $M_{1}(a)=a$ and $M_{s}M_{s'}=M_{ss'}$. Moreover,
$$M_{s}L^{r}(a_{i})=M_{s}(a_{i+r})=a_{sr+si}=L^{sr}(a_{si})=L^{sr}M_{s}(a_{i}).$$
\end{definition}

Therefore we have a transformation group $G=\{L^{r}M_{s}: r\in \mathbb{Z}_{n}, s\in \mathbb{Z}_{n}^{*}\}$ on $\sum(n)$ with relations $L^{0}=M_{1}=I_{d}$, $L^{r}L^{r'}=L^{r+r'}$, $M_{s}M_{s'}=M_{ss'}$, $M_{s}L^{r}=L^{sr}M_{s}$.

As another consequence, we show that $G$ is a transformation group on $\sum'(n)$.

\begin{corollary}\label{co4}
(1). If $a\in \sum'(n)$, then $\sigma(\overline{a})\in \sum'(n)$ and $\sigma(a)\in \sum'(n)$ for all $\sigma\in G$.

(2). For $a, b\in \sum'(n)$, interleaved sequences $w(a, b)$, $w(a, \overline{b})$ and $w(\sigma(a), \sigma(b)) (\sigma\in G)$ have the same linear complexity.
\end{corollary}

\begin{proof}
(1). Let $a=(a_i)_{i=0}^{n-1}$. By $a\in \sum'(n)$ we know that
$$A_{a}(\tau)=\sum_{i=0}^{n-1}(-1)^{a_{i}+a_{i+\tau}}=-1, ~~\textrm{for all}~~ 1\leq\tau\leq n-1.$$
Then $A_{\overline{a}}(\tau)=\sum_{i=0}^{n-1}(-1)^{1-a_{i}+1-a_{i+\tau}}
=\sum_{i=0}^{n-1}(-1)^{a_{i}+a_{i+\tau}}=A_{a}(\tau)=-1$. Therefore $\overline{a}\in \sum'(n)$.

For $r\in \mathbb{Z}_{n}$,
$$A_{L^{r}(a)}(\tau)=\sum_{i=0}^{n-1}(-1)^{a_{i+r}+a_{i+r+\tau}}
=\sum_{j=0}^{n-1}(-1)^{a_{j}+a_{j+\tau}}=A_{a}(\tau)=-1.$$
Therefore $L^{r}(a)\in \sum'(n)$. For $s\in \mathbb{Z}_{n}^{*}$, $1\leq\tau\leq n-1$,
$$A_{M_{s}(a)}(\tau)=\sum_{i=0}^{n-1}(-1)^{a_{is}+a_{s(\tau+i)}}
=\sum_{j=0}^{n-1}(-1)^{a_{j}+a_{j+s\tau}}=A_{a}(s\tau)=-1.$$
Therefore $M_{s}(a)\in \sum'(n)$. Since $G$ is generated by $L^{r}(0\leq r\leq n-1)$ and $M_{s}(s\in \mathbb{Z}_{n}^{*})$, we know that $\sigma(a)\in \sum'(n)$ for any $\sigma\in G$.

(2). Let $a=(a_i)_{i=0}^{n-1}$ and $b=(b_i)_{i=0}^{n-1}$. Then

$$S_{\overline{b}}(x)=\sum_{i=0}^{n-1}(1-b_{i})x^{i}=\frac{x^{n}-1}{x-1}+S_{b}(x)\in \mathbb{F}_{2}[x].$$
We get
$$\gcd \left(S_{a}(x), S_{\overline{b}}(x), \frac{x^{n}-1}{x-1}\right)=\gcd \left(S_{a}(x), S_{b}(x)+\frac{x^{n}-1}{x-1}, \frac{x^{n}-1}{x-1}\right)=\gcd \left(S_{a}(x), S_{b}(x), \frac{x^{n}-1}{x-1}\right),$$
$$\gcd \left(S_{a}(x)+S_{\overline{b}}(x), \frac{x^{n}-1}{x-1}\right)=\gcd \left(S_{a}(x)+S_{b}(x)+\frac{x^{n}-1}{x-1}, \frac{x^{n}-1}{x-1}\right)=\gcd \left(S_{a}(x)+S_{b}(x), \frac{x^{n}-1}{x-1}\right).$$
By Theorem \ref{th2}, $w(a, b)$ and $w(a, \overline{b})$ have the same linear complexity.

Next, for $r\in \mathbb{Z}_{n}, s\in \mathbb{Z}_{n}^{*}$
$$S_{L^{r}(b)}(x)=\sum_{i=0}^{n-1}b_{i+r}x^{i}=x^{-r}\sum_{j=0}^{n-1}b_{j}x^{j}\equiv x^{-r}S_{b}(x)\pmod {x^{n}-1}$$
$$S_{L^{r}(a)}(x)\equiv x^{-r}S_{a}(x)\pmod {x^{n}-1}$$
$$S_{M_{s}(b)}(x)=\sum_{i=0}^{n-1}b_{si}x^{i}\equiv\sum_{j=0}^{n-1}b_{j}x^{js^{*}}\equiv S_{b}(x^{s^{*}})\pmod {x^{n}-1}$$
$$S_{M_{s}(a)}(x)\equiv S_{a}(x^{s^{*}})\pmod {x^{n}-1}$$
where $s^{*}\in \mathbb{Z}$, $s^{*}s\equiv 1\pmod n$. We get
$$\gcd \left(S_{L^{r}(a)}(x), S_{L^{r}(b)}(x), \frac{x^{n}-1}{x-1}\right)=\gcd \left(x^{-r}S_{a}(x), x^{-r}S_{b}(x), \frac{x^{n}-1}{x-1}\right)=\gcd \left(S_{a}(x), S_{b}(x), \frac{x^{n}-1}{x-1}\right),$$
$$\gcd \left(S_{L^{r}(a)+L^{r}(b)}(x), \frac{x^{n}-1}{x-1}\right)=\gcd \left(x^{-r}S_{a+b}(x), \frac{x^{n}-1}{x-1}\right)=\gcd \left(S_{a+b}(x), \frac{x^{n}-1}{x-1}\right).$$
By Theorem \ref{th2}, $w(a, b)$ and $w(L^{r}(a), L^{r}(b))$ have the same linear complexity.\\
Similarly,
$$\gcd \left(S_{M_{s}(a)}(x), S_{M_{s}(b)}(x), \frac{x^{n}-1}{x-1}\right)=\gcd \left(S_{a}(x^{s^{*}}), S_{b}(x^{s^{*}}), \frac{x^{n}-1}{x-1}\right),$$
$$\gcd \left(S_{M_{s}(a)+M_{s}(b)}(x), \frac{x^{n}-1}{x-1}\right)=\gcd \left(S_{a+b}(x^{s^{*}}), \frac{x^{n}-1}{x-1}\right).$$

Let $\theta$ be a primitive root of 1 in an extension field of $\mathbb{F}_{2}$. Then for $1\leq \lambda\leq n-1$, $x=\theta^{\lambda}$ is a zero of $S_{a}(x^{s^{*}})(S_{b}(x^{s^{*}}), ~\textrm{or} ~S_{a+b}(x^{s^{*}}))\Longleftrightarrow \theta^{\lambda s^{*}}$ is a zero of $S_{a}(x)(S_{b}(x)~\textrm{or} ~S_{a+b}(x))$. By $s\in \mathbb{Z}_{n}^{*}$ we know that $\lambda\mapsto \lambda s$ is a permutation on $\mathbb{Z}_{n}^{*}$ . Therefore $w(a, b)$, $w(L^{r}(a), L^{r}(b))$($r\in \mathbb{Z}_{n}$) and $w(M_{s}(a), M_{s}(b))$($s\in \mathbb{Z}_{n}^{*}$) have the same linear complexity. Since the group $G$ is generated by such $L^{r}$ and $M_{s}$. We completes the proof of Corollary \ref{co4}.
\end{proof}

\section{ Construction of Interleaved Sequences with Maximum Linear Complexity $2n+2$}\label{sec5}
 In this section we construct several series of interleaved sequences $w=w(a,b)$ having the maximum linear complexity $LC(w)=2n+2$ by known binary sequences $a,b\in \sum'(n)$. For any $a\in \sum'(n)$, by Corollary $4$ we know that $\bar{a}, \sigma(a)\in \sum'(n)$ $(\sigma\in G)$. For any $a,b\in \sum'(n)$ we consider $LC(w)$ for $w=w(a', b')$ where $a'=\bar{a}$ or $\sigma(a)$, $b'=\bar{b}$ or $\tau(b)$ for $\sigma, \tau \in G$. By Corollary $4$,
 $$LC(w(\bar{a}, b))=LC(w(a,b)),\quad LC(w(\sigma(a), \tau(b)))=LC(w(a, \sigma^{-1}\tau(b))).$$
Therefore we may assume $a'=a$ and $b'=\sigma(b)$ $(\sigma\in G)$ without loss of generality.

The first known series of sequences in $\sum'(n)$ is $m$-sequences. Let $q=2^{l}(l\geq 2)$, $n=q-1$, $\gamma$ be a primitive element of $\mathbb{F}_{q}$, namely $\mathbb{F}^{*}_{q}=\langle\gamma\rangle=\{\gamma^{i}:0\leq i \leq n-1\}$. Let $\alpha \in \mathbb{F}^{*}_{q}$, $T:\mathbb{F}_{q}\rightarrow \mathbb{F}_{2}$ be the trace mapping. The $m$-sequence with period $n=q-1$ is $a=a(\alpha, \gamma)=\{a_{i}\}^{n-1}_{i=0}$ defined by
 $$a_{i}=a_{i}(\alpha, \gamma)=T(\alpha\gamma^{i})\quad (0\leq i\leq n-1).$$

Let $g_{a}(x)=\gcd \left(S_{a}(x), \frac{x^{n}-1}{x-1}\right)$, where $S_{a}(x)=\sum\limits_{i=0}^{n-1} a_{i}x^{i}\in \mathbb{F}_{2}[x]$. It is well known that
$$g_{a}(x)=\gcd \left(S_{a}(x), \frac{x^{n}-1}{x-1}\right)=\frac{x^{n}-1}{(x-1)m_{\gamma}(x)},$$
where $m_{\gamma}(x)$ is a primitive polynomial of degree $l$ (the minimal polynomial of $\gamma^{-1}$ over $\mathbb{F}_{2}$). Therefore $\deg g_{a}(x)=n-l-1$. For any $t\in \mathbb{Z}^{*}_{n}$, $\gamma'=\gamma^{t}$ is also a primitive element of $\mathbb{F}_{q}$. Let $b=(b_{i}=T(\alpha'{\gamma'}^{i}))_{i=0}^{n-1}$ be another $m$-sequence $(\alpha'\in \mathbb{F}^{*}_{q})$. Then $g_{b}(x)=\gcd \left(S_{b}(x), \frac{x^{n}-1}{x-1}\right)=\frac{x^{n}-1}{(x-1)m_{\gamma'}(x)}$ when $\alpha=\alpha'$, $b=M_t(a)$. When $m_{\gamma}(x)=m_{\gamma'}(x)$, $a$ and $b$ are shift-equivalent. From these facts we can see that for $0\leq r\leq n-1$,
\begin{align*}
& \gcd \left(S_{a}(x),\ S_{L^{r}(b)}(x),\ \frac{x^{n}-1}{x-1}\right)=\gcd \left(S_{a+L^{r}(b)}(x),\ \frac{x^{n}-1}{x-1}\right)=g_{a+L^{r}(b)}(x)
\\
&=\left\{ \begin{array}{ll}
\frac{x^{n}-1}{(x-1)m_{\gamma}(x)},& \textrm{if $m_{\gamma}(x)=m_{\gamma'}(x)$, $a\neq L^{r}(b)$;}\\
\frac{x^{n}-1}{(x-1)m_{\gamma}(x)m_{\gamma'}(x)},& \textrm{otherwise}.
\end{array} \right.
\end{align*}
From Theorem \ref{th2} we get $LC(w)=2n+2-2\deg g_{a+L^{r}(b)}(x)=2l+4$ or $4l+4$ for $w=w(a,L^{r}(b))(L^{r}(b)\neq a)$. Comparing with the period $n=2^{l}-1$, the linear complexity of $w$ is too small.

Now we consider some other series of sequences in $\sum'(n)$.
\subsection{Legendre Sequences}
 Let $p\equiv3\pmod 4$ be a prime number. For $1\leq i\leq p-1$, $(\frac{i}{p})$ is the Legendre symbol. There are two binary Legendre sequences $\ell=(\ell_{i})_{i=0}^{p-1}$ and $\ell'=(\ell'_{i})_{i=0}^{p-1}$ with period $p$ defined by $\ell_{0}=0, \ell'_{0}=1, ~\text{and for} ~1\leq i\leq p-1$,
$$\ell_{i}=\ell'_{i}
 = \left\{ \begin{array}{ll}
1,& \textrm{if $(\frac{i}{p})=1$;}\\
0,& \textrm{if $(\frac{i}{p})=-1$}.
\end{array} \right.$$

It is known that sequences $\ell$ and $\ell'$ (and their complement, shift, sample sequences) belong to $\sum'(p)$ \cite{X. Tang2}. The linear complexity of  $\ell$ has been determined in \cite{C. Ding2}. There are only two sample sequences of $\ell$: $M_{s}(\ell)=\ell$ for $(\frac{s}{p})=1$ and $M_{s}(\ell)=\ell^{*}=(\ell^{*}_{i})_{i=0}^{p-1}$ for $(\frac{s}{p})=-1$, where $\ell^{*}_{0}=0, ~\text{and for}~ 1\leq i\leq p-1$,
$$\ell^{*}_{i}=\ell_{si}=1-\ell_{i}
 = \left\{ \begin{array}{ll}
1,& \textrm{if $(\frac{i}{p})=-1$;}\\
0,& \textrm{if $(\frac{i}{p})=1$}.
\end{array} \right.$$
Since $\overline{\ell^{\ast}}=\ell'$, by Corollary \ref{co4}(2), it is reduced to consider $w=w(\ell,b)$ for $b=L^{r}(\ell)$ or $L^{r}(\ell')$.

\begin{theorem}\label{th5} Let $p\equiv3\ (mod\ 4)$ be a prime number, $p\geq7$, $\ell$ be the Legendre sequence with period $p$, $b=L^{r}(\ell')$, $1\leq r\leq p-1$. Then for interleaved sequence $w=w(\ell, b)$, $LC(w)=2p+2$.
\end{theorem}
\begin{proof} We have $S_{\ell}(x)=\sum\limits_{\substack{i=1 \\ (\frac{i}{p})=1}}^{p-1}x^{i}\in \mathbb{F}_{2}[x]$, $S_{\ell'}(x)=1+S_{\ell}(x)$ and $S_{b}(x)\equiv x^{p-r}S_{\ell'}(x)\pmod {x^{p}-1}$. Therefore $S_{\ell+b}(x)\equiv S_{\ell}(x)+x^{p-r}(S_{\ell}(x)+1)\pmod {x^{p}-1}$.

Let $\theta$ be a $p$-th primitive root of $1$ in an extension field $F$ of $\mathbb{F}_{2}$. We know from Theorem \ref{th2} that
\begin{align*}
LC(w)=2p+2& \Longleftrightarrow \gcd \left(S_{\ell}(x)+x^{p-r}(S_{\ell}(x)+1),\ \frac{x^{p}-1}{x-1}\right)=1\\
& \Longleftrightarrow \text{For each}\ \lambda, 1\leq \lambda\leq p-1,\  S_{\ell}(\theta^{\lambda})+\theta^{-r\lambda}(S_{\ell}(\theta^{\lambda})+1)\neq0.
\end{align*}
We consider two cases on $p\equiv3 \pmod4$.

\begin{case}\label{ca1} $p\equiv7\pmod 8$. In this case $(\frac{2}{p})=1$ and for $1\leq \lambda \leq p-1$ we have (in $F$)
$$S_{\ell}(\theta^{\lambda})^{2}=\left(\sum\limits_{(\frac{i}{p})=1}\theta^{\lambda i}\right)^{2}=\sum\limits_{(\frac{i}{p})=1}\theta^{2\lambda i}=\sum\limits_{(\frac{j}{p})=1}\theta^{\lambda j}=S_{\ell}(\theta^{\lambda})$$
which implies that $S_{\ell}(\theta^{\lambda})=0$ or $1$. Then
\begin{align*}
S_{\ell}(\theta^{\lambda})+\theta^{-r\lambda}(S_{\ell}(\theta^{\lambda})+1)
 = \left\{ \begin{array}{ll}
1,& \textrm{if\ $S_{\ell}(\theta^{\lambda})=1$;}\\
\theta^{-r\lambda}\neq 0,& \textrm{if\ $S_{\ell}(\theta^{\lambda})=0$}.
\end{array} \right.
\end{align*}
Therefore $LC(w)=2p+2$.
\end{case}

\begin{case}\label{ca2}
$p\equiv3\pmod8$. Then $(\frac{2}{p})=-1$ and for $1\leq \lambda\leq p-1$,
$$S_{\ell}(\theta^{\lambda})^{2}=\sum\limits_{(\frac{i}{p})=1}\theta^{2\lambda i}=\sum\limits_{(\frac{j}{p})=-1}\theta^{j\lambda}=\sum\limits_{(\frac{j}{p})=1}\theta^{j\lambda}+1+\sum\limits_{j=0}^{p-1}\theta^{j\lambda}=S_{\ell}(\theta^{\lambda})+1$$
which implies that $S_{\ell}(\theta^{\lambda})=w\in \mathbb{F}_{4}$ where $w^{2}=w+1$. Then
$$S_{\ell}(\theta^{\lambda})+\theta^{-r\lambda}(S_{\ell}(\theta^{\lambda})+1)=w+\theta^{-r\lambda}w^{2}.$$
If $w+\theta^{-r\lambda}w^{2}=0$, then $\theta^{-r\lambda}=w^{2}$. From $1\leq r,  \lambda\leq p-1$ we know that the order of $\theta^{-r\lambda}$ is $p$. On the other hand, the order of $w^{2}$ is $3$. We get a contradiction by assumption $p\neq3$. Therefore $LC(w)=2p+2$.
\end{case}
\end{proof}

\subsection{Hall's Sequences}
 Let $p=4x^{2}+27$ be a prime number, $\mathbb{F}^{*}_{p}=\langle g\rangle$ where $g$ is a primitive root $(\bmod p)$, $p-1=6f$, $D=\langle g^{6}\rangle$. Then $$D_{\lambda}=g^{\lambda}D\ \  (0\leq \lambda\leq5)$$ are the cyclotomic classes of order 6 in $\mathbb{F}_{p}$. The Hall's sequence with period $p$ is the binary sequence $h=(h_{i})^{p-1}_{i=0}$ defined by for $0\leq i\leq p-1$,
\begin{align*}
 h_{i}= \left\{ \begin{array}{ll}
1, & \textrm{if $i\in D_{0}\cup D_{1}\cup D_{3}$;}\\
0, & \textrm{otherwise}.
\end{array} \right.
\end{align*}

It is proved that $h$ and its complement, shift, sample sequences belong to $\sum'(p)$. It is easy to see that if $s$ and $s'$ in the same class $D_{\lambda}$, then $M_{s}(h)=M_{s'}(h)=M_{g^{\lambda}}(h)$. There are only six sample sequences of $h$: $M_{g^{\lambda}}(h) (0\leq \lambda\leq 5)$. By Corollary $4(2)$, it is reduced to consider $w=w(h, b)$ for $b=L^{r}M_{g^{s}}(h)$, $0\leq r\leq p-1$, $0\leq s\leq5$.

In order to compute the linear complexity of such sequence $w(h,b)$, we need some preparations. Define the following ``Gauss periods" of order $6$ on $\mathbb{F}_{p}$ by
 $$\eta_{\lambda}(x)=\sum\limits_{\substack{i=1 \\ i\in D_{\lambda}}}^{p-1}x^{i}\in \mathbb{F}_{2}[x]\ (0\leq \lambda\leq5),\quad \xi_{\lambda}(x)=\eta_{\lambda}(x)+\eta_{\lambda+3}(x)\ (0\leq\lambda\leq2).$$
 Then $S_{h}(x)=\eta_{0}(x)+\eta_{1}(x)+\eta_{3}(x)$. Let $\theta$ be a $p$-th primitive root of $1$ in some extension field $F$ of $\mathbb{F}_{2}$. Then for each $j$, $1\leq j\leq p-1$, $\theta^{j}$ is also a $p$-th primitive root of $1$ and $\eta_{\lambda}(\theta^{j}), \xi_{\lambda}(\theta^{j})\in F$.

  We have the following facts. Most of them have been proved in \cite{J-H. Kim}.
\begin{fact}
From $-1=g^{\frac{p-1}{2}}=g^{3f}$ we know $f$ is odd which implies that $-1\in D_{3}$. From
\begin{align*}
\left(\frac{2}{p}\right)
 = \left\{ \begin{array}{ll}
1,& \textrm{if $p\equiv7\pmod8$;}\\
-1,& \textrm{if $p\equiv3\pmod8$}.
\end{array} \right.
\end{align*}
and $2$ is a cubic residue $\pmod p$ (\cite{IR}, Theorem 5, p.36]) we know that $2\in D_{0}$ for $p\equiv7\pmod 8$, and $2\in D_{3}$ for $p\equiv3\pmod8$.
\end{fact}

\begin{fact}
(\cite{J-H. Kim}) For $S_{h}(x)=\sum\limits_{i=0}^{p-1}h_{i}x^{i}$ and $\beta=\theta^{j}\ (1\leq j\leq p-1)$,
\begin{align*}
S_{h}(\beta)S_{h}(\beta^{-1})
 = \left\{ \begin{array}{ll}
0,& \textrm{if $p\equiv7\pmod8$;}\\
1,& \textrm{if $p\equiv3\pmod8$}.
\end{array} \right.
\end{align*}
\end{fact}

\begin{fact} For $\beta=\theta^{j}\ (1\leq j\leq p-1)$, $\xi_{\lambda}(\beta)\in \mathbb{F}_{2}=\{0,1\}$. Moreover, there exists $\beta$ such that $\xi_{1}(\beta)=1$, $\xi_{0}(\beta)=\xi_{2}(\beta)=0$. From now on we take $\theta$ to be this $\beta$. Namely, we assume that $\xi_{1}(\theta)=1$,
$\xi_{0}(\theta)=\xi_{2}(\theta)=0$.

This result has been proved for case $p\equiv 7 \pmod8$ in \cite{J-H. Kim}. Since $2\in D_3$, then for $0\leq \lambda \leq 2$ we have $\xi_{\lambda}(\beta)^2=\xi_{\lambda}(\beta)$ which implies that $\xi_{\lambda}\in \{0,1 \}$. It can be seen that the proof for
$p\equiv 7\pmod 8$ also works for case $p\equiv3\pmod8$.
\end{fact}

\begin{fact} If $p\equiv7\pmod8$, then $\eta_{1}(\theta)=\eta_{2}(\theta)=\eta_{5}(\theta)=1$, $\eta_{0}(\theta)=\eta_{3}(\theta)=\eta_{4}(\theta)=0$ and
\begin{align*}
S_{h}(\theta^{j})
 = \left\{ \begin{array}{ll}
1,& \textrm{if $j\in D_{0}$;}\\
0,& \textrm{if $j\in \mathbb{F}^{*}_p\setminus D_{0}$}.
\end{array} \right.
\end{align*}
If $p\equiv3\pmod8$ then $\eta_{1}(\theta)=w$, $\eta_{4}(\theta)=w^{2}$ or $\eta_{1}(\theta)=w^{2}$, $\eta_{4}(\theta)=w$, where
$w\in \mathbb{F}_{4}$, $w^{2}+w+1=0$, and $\eta_{0}(\theta)=\eta_{3}(\theta)=1$, $\eta_{2}(\theta)=\eta_{5}(\theta)=0$,
\begin{align*}
S_{h}(\theta^{j})
 = \left\{ \begin{array}{ll}
w,& \textrm{if $j\in D_{0}$;}\\
w^{2},& \textrm{if $j\in D_{3}$;}\\
1,& \textrm{if $j\in \mathbb{F}^{*}_{p}\setminus (D_{0}\cup D_{3})$}.
\end{array} \right.
or \ \ S_{h}(\theta^{j})
 = \left\{ \begin{array}{ll}
w^{2},& \textrm{if $j\in D_{0}$;}\\
w,& \textrm{if $j\in D_{3}$;}\\
1,& \textrm{if $j\in \mathbb{F}^{*}_{p}\setminus (D_{0}\cup D_{3})$}.
\end{array} \right.
\end{align*}
\end{fact}
For $p\equiv7\pmod8$, the proof is given in \cite{J-H. Kim}. Now we show the case $p\equiv3\pmod8$ by similar method. By Fact (2) we know that $S_{h}(\theta^{j})S_{h}(\theta^{-j})=1\ (1\leq j\leq 5)$, and by Fact (3), $\xi_{1}(\theta)=1$, $\xi_{0}(\theta)=\xi_{2}(\theta)=0$ where $\xi_{\lambda}(\theta)=\eta_{\lambda}(\theta)+\eta_{\lambda+3}(\theta)$. We have
$$S_{h}(\theta)=\eta_{0}(\theta)+\eta_{1}(\theta)+\eta_{3}(\theta)=\eta_{1}(\theta)+\xi_{0}(\theta)=\eta_{1}(\theta),$$
\begin{align*}
S_{h}(\theta^{-1})& =\eta_{3}(\theta)+\eta_{4}(\theta)+\eta_{0}(\theta)\ (\text{since} -1\in D_{3})\\
& =\eta_{4}(\theta).
\end{align*}
Therefore $S_{h}(\theta)+S_{h}(\theta^{-1})=\eta_{1}(\theta)+\eta_{4}(\theta)=\xi_{1}(\theta)=1$. Then by $S_{h}(\theta)S_{h}(\theta^{-1})=1$, we get $S_{h}(\theta)=w$, $S_{h}(\theta^{-1})=w^{2}$ or $S_{h}(\theta)=w^{2}$, $S_{h}(\theta^{-1})=w$. Therefore $\eta_{1}(\theta)=S_{h}(\theta)=w$, $\eta_{4}(\theta)=S_{h}(\theta^{-1})=w^{2}$ or $\eta_{1}(\theta)=S_{h}(\theta)=w^{2}$, $\eta_{4}(\theta)=S_{h}(\theta^{-1})=w$.

\noindent Next we take $j\in D_{1}$. Then
$$S_{h}(\theta^{j})=\sum\limits_{i\in D_{0}\cup D_{1}\cup D_{3}}\theta^{ji}=\eta_{1}(\theta)+\eta_{2}(\theta)+\eta_{4}(\theta)=\eta_{2}(\theta)+1,$$
\begin{align*}
S_{h}(\theta^{-j})& =\eta_{4}(\theta)+\eta_{5}(\theta)+\eta_{1}(\theta)\ (\text{since} -1\in D_{3}, \text{and} -j\in D_{4})\\
& =\eta_{5}(\theta)+1.
\end{align*}
Therefore $S_{h}(\theta^{j})+S_{h}(\theta^{-j})=\xi_{2}(\theta)=0$. Then by $S_{h}(\theta^{j})S_{h}(\theta^{-j})=1$ we get $S_{h}(\theta^{j})=S_{h}(\theta^{-j})=1$ and $\eta_{2}(\theta)=\eta_{5}(\theta)=0$. At last we take $j\in D_{2}$, then
$$S_h(\theta^{j})=\eta_{2}(\theta)+\eta_{3}(\theta)+\eta_{5}(\theta)=\eta_{3}(\theta),$$
$$S_h(\theta^{-j})=\eta_{5}(\theta)+\eta_{0}(\theta)+\eta_{2}(\theta)=\eta_{0}(\theta).$$
\noindent Therefore $S_h(\theta^{j})+S_h(\theta^{-j})=\eta_{0}(\theta)+\eta_{3}(\theta)=0$, and $S_{h}(\theta^{j})S_{h}(\theta^{-j})=1$ we get $S_h(\theta^{j})=S_h(\theta^{-j})=1$ and $\eta_{0}(\theta)=\eta_{3}(\theta)=1$. From $-1\in D_{3}$, we know that if $j\in D_{\lambda}$, then $S_{h}(\theta^{-j})=S_{h}(\theta^{t}), t\in D_{\lambda+3}$. Therefore we have $S_{h}(\theta^{j})=w, w^{2}, 1$ or $S_{h}(\theta^{j})=w^{2}, w, 1$ for $j\in D_{0}, D_{3}$ and $\mathbb{F}^{*}_{p}\setminus (D_{0}\cup D_{3})$ respectively.

Fact $(4)$ gives the values of $\eta_{\lambda}(\theta)$ and $S_{h}(\theta^{j})$ for $0\leq \lambda\leq 5$ and $1\leq j\leq p-1$ as shown in the first two lines of Table (\uppercase\expandafter{\romannumeral1}) and Table (\uppercase\expandafter{\romannumeral2}). From these values we can compute the linear complexity of $w=w(h,\sigma(h))$ for all $\sigma\in G$.
\begin{table*}[h]
\centering
\caption{$p\equiv7\pmod8$}
  \setlength{\tabcolsep}{3mm}
	\begin{tabular}{cc|cccccc}
		
		& $\lambda$ & 0 & 1 & 2 & 3 & 4 & 5	\\
		\hline
		& $\eta_{\lambda}(\theta)$ & 0 & 1 & 1 & 0 & 0 & 1 \\
        \hline
		$j\in D_{\lambda}$ & $S_{h}(\theta^{j})$ & 1 & 0 & 0 & 0 & 0 & 0\\
		 & $S_{l}(\theta^{j})$ & 1 & 0 & 1 & 0 & 1 & 0\\
	\end{tabular}
\end{table*}
\begin{table*}[h]
\centering
\caption{$p\equiv3\pmod8$}
  \setlength{\tabcolsep}{3mm}
	\begin{tabular}{cc|cccccc}
		
		& $\lambda$ & 0 & 1 & 2 & 3 & 4 & 5	\\
		\hline
		& $\eta_{\lambda}(\theta)$ & 1 & $w$ & 0 & 1 & $w^{2}$ & 0 \\
        \hline
		$j\in D_{\lambda}$ & $S_{h}(\theta^{j})$ & $w$ & 1 & 1 & $w^{2}$ & 1 & 1\\
		 & $S_{l}(\theta^{j})$ & $w$ & $w^{2}$ & $w$ & $w^{2}$ & $w$ & $w^{2}$\\
	\end{tabular}
\end{table*}
\begin{table*}[h]
\centering
\setlength{\tabcolsep}{3mm}
	\begin{tabular}{cc|cccccc}
		
		& $\lambda$ & 0 & 1 & 2 & 3 & 4 & 5	\\
		\hline
		& $\eta_{\lambda}(\theta)$ & 1 & $w^{2}$ & 0 & 1 & $w$ & 0 \\
        \hline
		$j\in D_{\lambda}$ & $S_{h}(\theta^{j})$ & $w^{2}$ & 1 & 1 & $w$ & 1 & 1\\
		 & $S_{l}(\theta^{j})$ & $w^{2}$ & $w$ & $w^{2}$ & $w$ & $w^{2}$ & $w$\\
	\end{tabular}
\end{table*}

\begin{example} Suppose that $p=4x^{2}+27$. Compute the linear complexity of $w=w(h,b)$, $b=L^{2}M_{j}(h)$, $j\in D_{4}$. We have $\eta_{\lambda}(\theta^{j})=\eta_{\lambda+4}(\theta)$, and
\begin{align*}
S_{b}(x)& =\sum\limits_{i=0}^{p-1}h_{ij+2}x^{i}\equiv\sum\limits_{k=0}^{p-1}h_{k+2}x^{kj^{*}}\pmod{ x^{p}-1} \quad (j^{*}j\equiv1(\bmod p), ~\text{and}~ j^{*}\in D_{-4}=D_{2})\\
& \equiv x^{-2j^{*}}S_{h}(x^{j^{*}})\pmod{x^{p}-1}.
\end{align*}
For each $\lambda \in D_{\mu}$, $S_{b}(\theta^{\lambda})=\theta^{\alpha(\lambda)}S_{h}(\theta^{\beta(\lambda)})$ where
$$\alpha(\lambda)\equiv-2j^{*}\lambda\not\equiv 0 \pmod p,\quad \beta(\lambda)=j^{*}\lambda\in D_{\mu+2}.$$

\setlength{\parindent}{1em} (\uppercase\expandafter{\romannumeral1}). For $p\equiv7\pmod8$, we know from Table (\uppercase\expandafter{\romannumeral1}) that for $\lambda\in D_{\mu}$,
$$S_{h}(\theta^{\lambda})=0\Longleftrightarrow1\leq \mu\leq 5.$$
$$S_{b}(\theta^{\lambda})=0\Longleftrightarrow \beta(\lambda)=\lambda j^{*}\not\in D_{0}\Longleftrightarrow \lambda\not\in D_{-2}=D_{4}\Longleftrightarrow \mu\not\equiv4 \pmod6.$$
Therefore
$$S_{h}(\theta^{\lambda})=S_{b}(\theta^{\lambda})=0\Longleftrightarrow \lambda\in D_{1}\cup D_{2}\cup D_{3}\cup D_{5}.$$
On the other hand, for $\lambda\in D_{1}\cup D_{2}\cup D_{3}\cup D_{5}$, $S_{h}(\theta^{\lambda})+S_{b}(\theta^{\lambda})=0$. For $\lambda\in D_{0}$, we have $\beta(\lambda)\in D_2$ and
$$S_{h}(\theta^{\lambda})+S_{b}(\theta^{\lambda})=1+\theta^{\alpha(\lambda)}\cdot0=1\neq0.$$

\noindent For $\lambda\in D_{4}$ we have $\beta(\lambda)\in D_{0}$ and $S_{h}(\theta^{\lambda})+S_{b}(\theta^{\lambda})=0+\theta^{\alpha(\lambda)}\cdot1\neq0$. Therefore
$$\deg \gcd\left(S_{h}(x), S_{b}(x), \frac{x^{p}+1}{x+1}\right)=\deg \gcd\left(S_{h}(x)+S_{b}(x), \frac{x^{p}+1}{x+1}\right)=|D_{1}\cup D_{2}\cup D_{3}\cup D_{5}|=4\cdot\frac{p-1}{6},$$
and by $Theorem\ 2$, $LC(w)=2p+2-\frac{4(p-1)}{3}=\frac{1}{3}(2p+10)$.

\setlength{\parindent}{1em} (\uppercase\expandafter{\romannumeral2}). For $p\equiv3\pmod8$, from Table (\uppercase\expandafter{\romannumeral2}) we know that $S_{h}(\theta^{\lambda})\neq0$ for all $1\leq \lambda\leq p-1$. Therefore $\deg \gcd\left(S_{h}(x), S_{b}(x), \frac{x^{p}+1}{x+1}\right)=0$. On the other hand, $S_{h}(\theta^{\lambda})+S_{b}(\theta^{\lambda})=0\Longleftrightarrow S_{h}(\theta^{\lambda})=\theta^{\alpha(\lambda)}S_{h}(\theta^{\beta(\lambda)})$.

The order of $\theta^{\alpha(\lambda)}$ is $p=4x^{2}+27>3$. From Table (\uppercase\expandafter{\romannumeral2}) we know that the order of $S_{h}(\theta^{\lambda})S_{h}(\theta^{\beta(\lambda)})^{-1}$ is at most $3$ since it belongs to $\mathbb{F}^{*}_{4}=\{1, w, w^{2}\}$. This implies that $S_{h}(\theta^{\lambda})+S_{b}(\theta^{\lambda})\neq0$ for all $1\leq \lambda\leq p-1$. Therefore $\deg \gcd\left(S_{h}(x)+S_{b}(x), \frac{x^{p}+1}{x+1}\right)=0$, and $LC(w)=2p+2$.
\end{example}

Now we present a series of $w=w(h, \sigma(h))$ having the maximum linear complexity $2p+2$.

\begin{theorem}\label{th5} Let $p=4x^{2}+27\equiv3\pmod8$ (which means that $2|x)$, $h$ be the Hall sequence with period $p$, $b=\sigma(h)$, $\sigma=L^{r}M_{s}\in G$ where $1\leq r\leq p-1$, $s\in D_{\mu}$. Then $LC(w)=2p+2$ for $w=w(h,b)$.
\end{theorem}
\begin{proof}
From Corollary \ref{co3} we know that $LC(w)=2p+2$ if and only if $\gcd\left(S_{h+b}(x), \frac{x^{p}-1}{x-1}\right)=1$ in $\mathbb{F}_{2}[x]$ where
$$S_{h+b}(x)\equiv S_{h}(x)+x^{\alpha(\sigma)}S_{h}(x^{s'})\pmod {x^{p}-1},$$
and $\alpha(\sigma)=-rs'$, $s's\equiv1\pmod p$ so that $s'\in D_{-\mu}$.

Let $\theta$ be a primitive root of $1$ in an extension of $\mathbb{F}_{2}$. Then $\gcd\left(S_{h+b}(x), \frac{x^{p}-1}{x-1}\right)=1$ if and only if for all $j$, $1\leq j\leq p-1$,
$$0\neq S_{h+b}(\theta^{j})=S_{h}(\theta^{j})+\theta^{\alpha(\sigma)j}S_{h}(\theta^{s'j}).$$
From Table (\uppercase\expandafter{\romannumeral2}) we know that $S_{h}(\theta^{j})\in\{1, w, w^{2}\}$ for all $j$, $1\leq j\leq p-1$. If $S_{h+b}(\theta^{j})=0$ then $\theta^{\alpha(\sigma)j}=S_{h}(\theta^{j})/S_{h}(\theta^{s'j})\in\{1, w, w^{2}\}$. We get $\theta^{3j\alpha(\sigma)}=1$. From $p=4x^{2}+27>3$ and $3j\alpha(\sigma)=-3rs'j\not\equiv0\pmod p$ we know that $\theta^{3j\alpha(\sigma)}\neq1$. This contradiction shows that $\gcd\left(S_{h+b}(x), \frac{x^{p}-1}{x-1}\right)=1$ and $LC(w)=2p+2$.
\end{proof}
\begin{remark}By using Table (\uppercase\expandafter{\romannumeral1}), it can be proved that if $p\equiv7\pmod 8$ then $LC(w)<2p+2$ for $w=w(h, \sigma(h))$ and all $\sigma\in G$.
\end{remark}

Now we consider sequences interleaved by Legendre and Hall sequences. Recall that the Legendre sequence $\ell=(\ell_{i})_{i=0}^{p-1}$, $\ell'=(\ell'_{i})^{p-1}_{i=0}$ with period $p\equiv3\pmod 4$ are defined by $\ell_{0}=0$, $\ell_{0}'=1$, and for $1\leq i\leq p-1$,
\begin{align*}
\ell_{i}=\ell_{i}'
 = \left\{ \begin{array}{ll}
1,& \textrm{if $(\frac{i}{p})=1$;}\\
0,& \textrm{if $(\frac{i}{p})=-1$}.
\end{array} \right.
\end{align*}
\begin{theorem}Let $p=4x^{2}+27\equiv3\pmod 8$ be a prime number, $\ell$ and $\ell'$ be the Legendre sequences with period $p$, $h$ be the Hall sequence with period $p$. Then for $\sigma=L^{r}M_{s}\in G\ (0\leq r\leq p-1, 1\leq s\leq p-1)$ and $b=\sigma(h)$, we have

(1). $LC(w)=2p+2$ for $w=w(\ell, b)$ if $r\neq0$ or $``r=0$ and $(\frac{s}{p})=-1".$

(2). $LC(w)=2p+2$ for $w=w(\ell', b)$ if $r\neq0$ or $``r=0$ and $(\frac{s}{p})=1".$
\end{theorem}
\begin{proof}For $1\leq j\leq p-1$, $S_{\ell'}(\theta^{j})=S_{\ell}(\theta^{j})+1$ and
\begin{align*}
S_{\ell}(\theta^{j})
 = \left\{ \begin{array}{ll}
S_{\ell}(\theta)=\eta_{0}(\theta)+\eta_{2}(\theta)+\eta_{4}(\theta),& \textrm{if $(\frac{j}{p})=1\ (\textrm{namely}, j\in D_{0}\cup D_{2}\cup D_{4})$;}\\
\eta_{1}(\theta)+\eta_{3}(\theta)+\eta_{5}(\theta)=1+S_{\ell}(\theta),& \textrm{if $(\frac{j}{p})=-1\ (\textrm{namely}, j\in D_{1}\cup D_{3}\cup D_{5})$}.
\end{array} \right.
\end{align*}
From $\eta_{\lambda}(\theta)$ given in Fact (4) we can compute the values of $S_{\ell}(\theta^{j})\ (1\leq j\leq p-1)$ as shown in the last line of Table (\uppercase\expandafter{\romannumeral1}) and (\uppercase\expandafter{\romannumeral2}).
$$S_{\ell+b}(\theta^{j})=S_{\ell}(\theta^{j})+\theta^{-s'rj}S_{h}(\theta^{s'j}),\quad S_{\ell'+b}(\theta^{j})=1+S_{\ell}(\theta^{j})+\theta^{-s'rj}S_{h}(\theta^{s'j}),$$ where $s's\equiv1\pmod p$.
From Table (\uppercase\expandafter{\romannumeral2}) we can see that for $p\equiv3\pmod 8$, $S_{\ell}(\theta^{j})$ and $S_{\ell'}(\theta^{j})=S_{\ell}(\theta^{j})+1$ belong to $\{w, w^{2}\}$ for all $j$, $1\leq j\leq p-1$. Thus
$$S_{\ell+b}(\theta^{j})=0\Leftrightarrow \theta^{-s'rj}=S_{\ell}(\theta^{j})/S_{h}(\theta^{s'j})\in\{1, w, w^{2}\},$$
$$S_{\ell'+b}(\theta^{j})=0\Leftrightarrow \theta^{-s'rj}=(1+S_{\ell}(\theta^{j}))/S_{h}(\theta^{s'j})\in\{1, w, w^{2}\}.$$

If $r\neq0$, $p\nmid (-s'rj)$ and the order of $\theta^{-s'rj}$ is $p=4x^{2}+27>3$. Then from the above equivalent relations we know that $S_{\ell+b}(\theta^{j})\neq0$, $S_{\ell'+b}(\theta^{j})\neq0$ for all $1\leq j\leq p-1$. Therefore by Corollary \ref{co3} we have $LC(w)=2p+2$ for $w=w(\ell, b)$ and $w(\ell', b)$.

If $r=0$, for $(\frac{s'}{p})=(\frac{s}{p})=-1$, let $j\in D_{\mu}$, $s'j\in D_{\lambda}$, then $2\nmid (u+\lambda)$. From Table \uppercase\expandafter{\romannumeral2} we can see that
$$S_{\ell+b}(\theta^{j})=S_{\ell}(\theta^{j})+S_{h}(\theta^{s'j})\neq0,$$
for all $1\leq j\leq p-1$. Therefore $LC(w)=2p+2$ for $w=w(\ell, b)$. Similarly, it can be shown that if $r=0$ and $(\frac{s'}{p})=(\frac{s}{p})=1$, then $S_{\ell'+b}(\theta^{j})=S_{\ell}(\theta^{j})+1+S_{h}(\theta^{s'j})\neq0$ for all $1\leq j\leq p-1$. Therefore  by Corollary \ref{co3} we have $LC(w)=2p+2$ for $w=(\ell', b)$.
\end{proof}
\begin{remark}Using Table (\uppercase\expandafter{\romannumeral1}) we can verify that if $p\equiv7\pmod 8$ then $LC(w)<2p+2$ for $w=w(a, b)$ where $a=\ell$ or $\ell'$ and $b=\sigma(h)\ (\sigma\in G)$.
\end{remark}

\subsection{Twin-prime Sequences}
Let $n=pq$ where $p$ and $q=p+2$ are prime numbers. Then we have the partition
$$\mathbb{Z}_{n}=\{0\}\cup P\cup Q\cup \mathbb{Z}^{*}_{n},\quad \mathbb{Z}^{*}_{n}=D_{0}\cup D_{1}$$
where $P=\{p, 2p,\cdots, (q-1)p\}$, $Q=\{q, 2q,\cdots, (p-1)q\}$,
$$D_{0}=\{i\in \mathbb{Z}^{*}_{n}: \left(\frac{i}{p}\right)\left(\frac{i}{q}\right)=1\},\quad D_{1}=\{i\in \mathbb{Z}^{*}_{n}: \left(\frac{i}{p}\right)\left(\frac{i}{q}\right)=-1\}.$$
The twin-prime sequence $t=(t_{i})_{i=0}^{n-1}$ with period $n=pq$ is defined by $t_{0}=0$ and for $1\leq i\leq n-1$,
\begin{align*}
t_{i}
 = \left\{ \begin{array}{ll}
0,& \textrm{if $i\in Q\cup D_{0}$;}\\
1,& \textrm{if $i\in P\cup D_{1}$}.
\end{array} \right.
\end{align*}
It is known that the sequence $t$ and its complement, shift, sample sequences belong to $\sum'(n)$. The linear complexity of $t$ has been determined in \cite{C. Ding1}. The sequence $t$ has two samples: for $s\in \mathbb{Z}^{*}_{n}$,
\begin{align*}
M_{s}(t)= \left\{ \begin{array}{ll}
t, & \textrm{if $(\frac{s}{p})(\frac{s}{q})=1$;}\\
\tau(t), & \textrm{if $(\frac{s}{p})(\frac{s}{q})=-1$},
\end{array} \right.
\end{align*}
where $\tau(t)=(\tau(t)_{i})_{i=0}^{n-1}$ is defined by $\tau(t)_{0}=0$, for $1\leq i\leq n-1$,
\begin{align*}
\tau(t)_{i}
 = \left\{ \begin{array}{ll}
1,& \textrm{if $i\in P\cup D_{0}$;}\\
0,& \textrm{if $i\in Q\cup D_{1}$}.
\end{array} \right.
\end{align*}
By Corollary \ref{co4}, it is reduced to consider $LC(w)$ for $w=(a, L^r(\tau(b))$, where $a=t$ and $b=t$ or $\tau(t)$.

We have in $\mathbb{F}_{2}[x]$,
$$S_{\tau(t)}(x)=S_{t}(x)+\sum\limits_{i\in \mathbb{Z}^{*}_{n}}x^{i}=S_{t}(x)+\sum\limits_{i\in \mathbb{Z}_{n}}x^{i}+\sum\limits_{i=0}^{p-1}x^{qi}+\sum\limits_{j=0}^{q-1}x^{pj}+1=S_{t}(x)+\frac{x^{n}-1}{x-1}+\frac{x^{n}-1}{x^{q}-1}+\frac{x^{n}-1}{x^{p}-1}+1.$$

Now we compute $S_{t}(x)\pmod {\frac{x^{n}-1}{x-1}}$. We define, for $\varepsilon\in\{\pm1\}$
$$G_{p,\varepsilon}(x)=\sum_{\substack{i=1 \\ (\frac{i}{p})=\varepsilon}}^{p-1}x^{qi},\ \  G_{q,\varepsilon}(x)=\sum_{\substack{j=1 \\ (\frac{j}{q})=\varepsilon}}^{q-1}x^{pj}\in \mathbb{F}_2[x]$$
Then
$$G_{p,1}(x)+G_{p,-1}(x)=\sum\limits_{i=1}^{p-1}x^{qi}=1+\frac{x^{n}-1}{x^{q}-1},\quad G_{q,1}(x)+G_{q,-1}(x)=1+\frac{x^{n}-1}{x^{p}-1}.$$
\begin{lemma}\label{lm8}Let $n=pq$, where $p$ and $q=p+2$ are prime numbers. Then

(1). $S_{t}(x)\equiv G_{q,1}(x)(1+\frac{x^{n}-1}{x^{q}-1})+(G_{p,1}(x)+1)(1+\frac{x^{n}-1}{x^{p}-1})\pmod {x^{n}-1}$.

(2). Let $\theta$ be an n-th primitive root of 1 in an extension field of $\mathbb{F}_{2}$. Then for $1\leq j\leq n-1$,
\begin{align*}
S_{t}(\theta^{j})
 = \left\{ \begin{array}{ll}
\frac{p+1}{2},& \textrm{if $\gcd (j, n)>1$;}\\
G_{q,1}(\theta^{j})+G_{p,1}(\theta^{j})+1,& \textrm{if $\gcd (j, n)=1$}.
\end{array} \right.
\end{align*}

(3). When $\gcd (j, n)=1$,
\begin{align*}
G_{q,1}(\theta^{j})
 = \left\{ \begin{array}{ll}
0\ or\ 1,& \textrm{if $(\frac{2}{q})=1\ (namely, q\equiv\pm1\ (mod\ 8))$;}\\
w,& \textrm{if $(\frac{2}{q})=-1\ (q\equiv\pm3\ (mod\ 8))$},
\end{array} \right.
\end{align*}
\begin{align*}
G_{p,1}(\theta^{j})
 = \left\{ \begin{array}{ll}
0\ or\ 1,& \textrm{if $(\frac{2}{p})=1\ ( p\equiv\pm1\ (mod\ 8))$;}\\
w,& \textrm{if $(\frac{2}{p})=-1\ (p\equiv\pm3\ (mod\ 8))$},
\end{array} \right.
\end{align*}
where $w\in \mathbb{F}_{4}$, $w^{2}+w+1=0$.
\end{lemma}
\begin{proof}(1). By the definition of the sequence $t$,
\begin{align*}
S_{t}(x) &=\sum\limits_{j=1}^{q-1}x^{pj}+\sum_{\substack{\lambda\in \mathbb{Z}^{*}_{n} \\ (\frac{\lambda}{p})(\frac{\lambda}{q})=-1}}x^{\lambda}\ (\text{let}\  \lambda=pj+qi)\\
& \equiv1+\frac{x^{n}-1}{x^{p}-1}+\substack{{\sum\limits_{i=1}^{p-1}\sum\limits_{j=1}^{q-1}} \\ {(\frac{j}{q})(\frac{i}{p})=-(\frac{p}{q})(\frac{q}{p})}}x^{pj+qi}\ \pmod {x^{n}-1}.
\end{align*}
\noindent By the quadratic reciprocity law, $(\frac{p}{q})(\frac{q}{p})=(-1)^{\frac{p-1}{2}\frac{q-1}{2}}=(-1)^{\frac{p^{2}-1}{4}}=1$, we get
\begin{align}
S_{t}(x)& \equiv1+\frac{x^{n}-1}{x^{p}-1}+\substack{{\sum\limits_{i=1}^{p-1}\sum\limits_{j=1}^{q-1}} \\ \notag {(\frac{i}{p})(\frac{j}{q})=-1}}x^{pj}\cdot x^{qi} \pmod {x^{n}-1} \\ \notag
& \equiv1+\frac{x^{n}-1}{x^{p}-1}+G_{p,1}(x)G_{q,-1}(x)+G_{p,-1}(x)G_{q,1}(x) \pmod {x^{n}-1} \\ \notag
& \equiv1+\frac{x^{n}-1}{x^{p}-1}+G_{p,1}(x)(G_{q,1}(x)+1+\frac{x^{n}-1}{x^{p}-1})+(1+\frac{x^{n}-1}{x^{q}-1}+G_{p,1}(x))G_{q,1}(x) \pmod {x^{n}-1} \\
& \equiv G_{q,1}(x)(1+\frac{x^{n}-1}{x^{q}-1})+(G_{p,1}(x)+1)(1+\frac{x^{n}-1}{x^{p}-1}) \pmod {x^{n}-1}.\label{e2}
\end{align}
\noindent (2). If $p|j$, $j=p\lambda$, $1\leq \lambda \leq q-1$. By (\ref{e2}) we get
\begin{align*}
S_{t}(\theta^{j})& =S_{t}(\theta^{p\lambda})=G_{q,1}(\theta^{p\lambda})(1+p)+G_{p,1}(\theta^{p\lambda})+1=G_{p,1}(\theta^{p\lambda})+1\\
& =\sum_{\substack{i=1 \\ (\frac{i}{p})=1}}^{p-1} (\theta^{p\lambda})^{qi}+1=1+\sum_{\substack{i=1 \\ (\frac{i}{p})=1}}^{p-1}1=1+\frac{p-1}{2}=\frac{p+1}{2}.
\end{align*}
\noindent If $q|j$, $j=q\lambda$, $1\leq\lambda\leq p-1$, we have
$$S_{t}(\theta^{j})=G_{q,1}(\theta^{q\lambda})+(G_{p,1}(\theta^{q\lambda})+1)(1+q)=G_{q,1}(\theta^{q\lambda})=\frac{q-1}{2}=\frac{p+1}{2}.$$
At last, if $\gcd(j, pq)=1$, then $\frac{x^{n}-1}{x^{q}-1}|_{x=\theta^{j}}$ and $\frac{x^{n}-1}{x^{p}-1}|_{x=\theta^{j}}$ are zero, and $S_{t}(\theta^{j})=G_{q,1}(\theta^{j})+G_{p,1}(\theta^{j})+1$.

\noindent (3). If $(\frac{2}{q})=1$, then
$$G_{q,1}(x)^{2}=(\sum_{\substack{i=1 \\ (\frac{i}{q})=1}}^{q-1}x^{pi})^2\equiv(\sum_{\substack{i=1 \\ (\frac{i}{q})=1}}^{q-1}x^{2pi})\equiv(\sum_{\substack{\lambda=1 \\ (\frac{\lambda}{q})=1}}^{q-1}x^{p\lambda})= G_{q,1}(x) \pmod {x^{n}-1}.$$
Therefore $G_{q,1}(\theta^{j})^{2}=G_{q,1}(\theta^{j})$ which means that $G_{q,1}(\theta^{j})\in \mathbb{F}_{2}$. Similarly, if $(\frac{2}{p})=1$ we have $G_{p,1}(\theta^{j})\in \mathbb{F}_{2}$. If $(\frac{2}{q})=-1$, then
$$G_{q,1}(x)^{2}=(\sum_{\substack{i=1 \\ (\frac{i}{q})=1}}^{q-1}x^{2pi})\equiv(\sum_{\substack{\lambda=1 \\ (\frac{\lambda}{q})=1}}^{q-1}x^{p\lambda})=G_{q,-1}(x)=G_{q,1}(x)+1+\frac{x^{n}-1}{x^{p}-1} \pmod {x^{n}-1}.$$
Since $\gcd(j,n)=1$ we have $G_{q,1}(\theta^{j})^{2}=G_{q,1}(\theta^{j})+1$. Therefore $G_{q,1}(\theta^{j})=w$. Similarly, if $(\frac{2}{p})=-1$ we have $G_{p,1}(\theta^{j})=w$.
\end{proof}
\begin{theorem}
Let $n=pq$ where $p$ and $q=p+2$ are prime numbers, $t$ be the twin-prime sequence with period $n$. Then for $r\in \mathbb{Z}^{*}_{n}$

(1). If $p\equiv1 \pmod4$, then $LC(w)=2n+2$ for $w=w(t,L^{r}(t))$.

(2). If $p\equiv1 \pmod4$, then $LC(w)=2n+2$ for $w=w(t,L^{r}(\tau(b)))$, $b=t$.
\end{theorem}
\begin{proof}(1). From $S_{t+L^{r}(t)}(x)\equiv(1+x^{-r})S_{t}(x) \pmod {x^{n}-1}$ we get, for $1\leq j\leq n-1$, $S_{t+L^{r}(t)}(\theta^{j})=(1+\theta^{-rj})S_{t}(\theta^{j})$. From $\gcd(r,n)=1$ we know that $\theta^{-rj}\neq1$ and $1+\theta^{-rj}\neq0$. By assumption $p\equiv1 \pmod4$ and Lemma \ref{lm8}, we get $S_{t}(\theta^{j})=\frac{p+1}{2}\equiv1\in \mathbb{F}_{2}$ if $\gcd(j,n)>1$. Moreover, if $\gcd(j,n)=1$, then $S_{t}(\theta^{j})=G_{q,1}(\theta^{j})+G_{p,1}(\theta^{j})+1$. If $p\equiv1 \pmod8$, then $q=p+2\equiv3 \pmod8$. By Lemma \ref{lm8}, $G_{q,1}(\theta^{j})=w$, $G_{p,1}(\theta^{j})+1\in \mathbb{F}_{2}$. We get $S_{t}(\theta^{j})\neq0$. Similarly, if $p\equiv5 \pmod8$, then $q\equiv7 \pmod8$. We get $G_{p,1}(\theta^{j})=w$, $G_{q,1}(\theta^{j})+1\in \mathbb{F}_{2}$. We also have $S_{t}(\theta^{j})\neq0$. Therefore $(1+\theta^{-rj})S_{t}(\theta^{j})\neq0$ for all $1\leq j\leq n-1$, which means that $LC(w)=2n+2$ for $w=w(t,L^{r}(t))$ by Corollary \ref{co3}.

(2). From
$$
S_{t+L^{r}(\tau(t))}(x)=
S_{t}(x)+x^{-r}(S_{t}(x)+1+\frac{x^{n}-1}{x^{p}-1}+\frac{x^{n}-1}{x^{q}-1}+\frac{x^n-1}{x-1}),
$$
we know that for all $1\leq j\leq n-1$,
$$S_{t+L^{r}(\tau(t))}(\theta^{j})=S_{t}(\theta^{j})+\theta^{-rj}(S_{t}(\theta^{j})+\varepsilon),\ \varepsilon\in \mathbb{F}_{2}.$$
If $\varepsilon=0$, $S_{t+L^{r}(\tau(t))}(\theta^{j})=(1+\theta^{-rj})S_{t}(\theta^{j})$. By the proof of (1), we know that $S_{t+L^{r}(\tau(t))}(\theta^{j})\neq0$. If $\varepsilon=1$, $S_{t+L^{r}(\tau(t))}(\theta^{j})=S_{t}(\theta^{j})+\theta^{-rj}(S_{t}(\theta^{j})+1)$. If $S_{t}(\theta^{j})=0$ or 1, then $S_{t+L^{r}(\tau(t))}(\theta^{j})=\theta^{-rj}$ or 1. Otherwise, $S_{t}(\theta^{j})=w$ or $w^2$ by Lemma \ref{lm8} and $S_{t+L^{r}(t)}(\theta^{j})=w+\theta^{-rj}w^{2}$ or $w^{2}+\theta^{-rj}w$. From $1\leq j\leq n-1$ and $r\in \mathbb{Z}^{*}_{n}$ we know that $\theta^{-rj}\neq1$ so that the order of $\theta^{-rj}$ is at least $p>3$. If $S_{t+L^{r}(t)}(\theta^{j})=0$, then $\theta^{-rj}=w/w^{2}=w^{2}$ or $w^{2}/w=w$ and the order of $w^{2}$ or $w$ is 3. Therefore $LC(w)=2n+2$ for $w=(t,L^{r}(\tau(t)))$ by Corollary \ref{co3}.
\end{proof}
\section{Conclusion}
In this paper, we determine the linear complexity $LC(w)$ of the binary sequences $w=w(a,b)$ with period $4n$ interleaved by two binary sequences $a$ and $b$ with period $n\equiv3 \pmod4$ and having ideal autocorrelation. We present a general formula on $LC(w)$ and show that $LC(w)\leq 2n+2$ for all such interleaved sequence $w$. Then we present many series of $w=w(a, b)$ such that $LC(w)$ reaches the maximum values $2n+2$ by considering $a$ and $b$ being several classes of known binary sequences with ideal autocorrelation ($m$-sequences, Legendre, twin-prime and Hall sequences and their complement, shift and sample sequences). Many other series of binary sequences with ideal autocorrelation have been found (GMW, Kasami sequences and others). To determine the linear complexity of sequences $w(a, b)$ interleaved by new types $a$ and $b$ may need more technique.

As a final remark, we mention that for all binary sequences $w=w(a, b)$ interleaved by arbitrary binary sequences $a$ and $b$ with ideal autocorrelation, the 2-adic complexity of $w$ reaches the maximum value $\log_{2}(2^{4n}-1)$. This can be done easily by the method given in \cite{H. Xiong}.

\ifCLASSOPTIONcaptionsoff
  \newpage
\fi



%


\end{document}